\documentclass[letterpaper, 10 pt, conference]{ieeeconf}  

\usepackage{srcltx}
\IEEEoverridecommandlockouts                              

\usepackage{bm}
 \usepackage{amssymb,latexsym,amsfonts,amsmath,amscd} 
\usepackage{caption}
\usepackage{subcaption}
\usepackage{cite}    
\usepackage{float}
\usepackage{verbatim}

\usepackage{cite}
\usepackage{acronym}
 
\usepackage{graphicx}
\usepackage{paralist}
\newtheorem{definition}{Definition}
\newtheorem{theorem}{Theorem}
\newtheorem{lemma}{Lemma}

\newtheorem{example}{Example}

\newtheorem{problem}{Problem}

\acrodef{sa}[SA]{semiautomaton}
\acrodef{fsa}[FSA]{finite state automaton}
\acrodef{dfa}[DFA]{deterministic finite state automaton}
\acrodef{pta}[PTA]{prefix tree acceptor}
\acrodef{smt}[SMT]{Satisfiability Modulo Theories}
\acrodef{ltl}[LTL]{Linear temporal logic}
\acrodef{ts}[TS]{transition system}

\newcommand{\truev}{\mathsf{true}}
\newcommand{\falsev}{\mathsf{false}}

\acrodef{pomdp}[POMDP]{partially observable Markov decision process}

\acrodef{act}[\textsc{ACT}]{abstract counterexample tree}
\acrodef{cegar}[CEGAR]{counterexample guided abstraction refinement}
\providecommand{\abs}[1]{\lvert#1\rvert}

\title{Abstractions and sensor design in partial-information, reactive
controller synthesis}

\author{Jie Fu, Rayna Dimitrova and Ufuk Topcu
\thanks{This work was supported in part by the AFOSR (FA9550-12-1-0302) and ONR (N00014-13-1-0778).}
\thanks{The authors are with the University of Delaware (USA), Max
  Planck Institute for Software Systems (Germany), and the University
  of Pennsylvania (USA), respectively.}}

\begin{document}

\maketitle
\thispagestyle{empty}
\pagestyle{empty}

\begin{abstract}
  Automated synthesis of reactive control protocols from temporal
  logic specifications has recently attracted considerable attention
  in various applications in, for example, robotic motion planning,
  network management, and hardware design. An implicit and often
  unrealistic assumption in this past work is the availability of
  complete and precise sensing information during the execution of the
  controllers. In this paper, we use an abstraction procedure for
  systems with partial observation and propose a formalism to
  investigate effects of limitations in sensing. The abstraction
  procedure enables the existing synthesis methods with partial
  observation to be applicable and efficient for systems with infinite
  (or finite but large number of) states. This formalism enables us to
  systematically discover sensing modalities necessary in order to
  render the underlying synthesis problems feasible. 
  We use counterexamples, which witness unrealizability potentially
  due to the limitations in sensing and the coarseness in the abstract
  system, and interpolation-based techniques to refine the model and
  the sensing modalities, i.e., to identify new sensors to be
  included, in such synthesis problems. We demonstrate the method on
  examples from robotic motion planning.
\end{abstract}

\section{Introduction}
\label{sec:intro}

Automatically synthesizing reactive controllers with proofs of
correctness for given temporal logic specifications has
emerged as a methodology complementing post-design verification
efforts in building assurance in system operation. Its recent
applications include autonomous robots \cite{ram-hadas,itac}, hardware
design \cite{GR(1)case}, and vehicle management systems
\cite{vms0}. This increasing interest is partly due to both
theoretical advances \cite{pnueli-hardmis, piterman} and software
toolset developments
\cite{pnueli2010jtlv,bloem2010ratsy,wongpiromsarn2011tulip}.

An implicit and often unrealistic assumption in the past work on
reactive synthesis is the availability of complete and precise 
information during the execution of controllers. For example,
while navigating through a workspace, a robot rarely (if ever) has
global awareness about its surrounding dynamic environment and its
sensing of even its own configuration is imprecise. This paper takes
an initial step toward explicitly accounting for the effects of such
incompleteness and imperfectness in sensing (and other means through
which information is revealed to the controller at runtime). 

More specifically, we use an abstraction procedure for games with
partial observation \cite{rayna-report} and propose a formalism to
investigate the effects of limitations in sensing. The abstraction
reduces the size of the control synthesis problem with sensing
limitations by focusing on relevant properties of the control
objective and enables automatic synthesis for systems with potentially
large state spaces using the solutions for partial-information,
turn-based, temporal-logic games \cite{apt2011lectures,
    chatterjee2006algorithms}. Given unrealizable specifications,
where a potential cause for unrealizability is the lack of runtime
information, a simple question we investigate is
what new sensing modalities and with what precision shall be included
in order to render the underlying synthesis problem feasible. We focus on
particular safety type temporal logic specifications for which
counterexamples witness the unrealizability. Using such
counterexamples and interpolation-based techniques
\cite{cimatti2012smt}, the method searches for predicates to be
included in the abstraction. We interpret addition of such
newly discovered predicates as abstraction refinements as well as
adding new sensing modalities or increasing the precision of the
existing sensors.
Besides the partial-information, turn-based games (see
\cite{arnold2003games,Martin2006} in addition to the earlier
references mentioned)  the problem we study in this paper has
similarities with the partially observable Markov decision processes
\cite{kaelbling1998planning,pineau2002high,kurniawati2008sarsop}. The
main deviation in the formalism we employ is the inclusion of a second
player which represents a dynamic, possibly adversarial environment,
particularly well suited for reactive synthesis in a number of
applications, for example, autonomous navigation.

The rest of the paper is organized as follows. We begin with an
overview of the setup, problem, and solution approach. In section
\ref{sec:problemformulation}, we discuss some preliminaries as they
build toward a formal statement of the problem. The solution approach
is detailed in the following two sections in which first an
abstraction procedure and then refinements in abstractions based on
counterexamples are presented. This presentation partly follows the
development in \cite{rayna-report}. Section \ref{subsec:sensorconf}
gives an interpretation of the results in the reconfiguration of
sensing modalities and section \ref{subsec:case} is on a case
study. Throughout the paper, we consider motivating and running
examples loosely from the context of autonomous robotic motion
planning subject to temporal logic specifications.

\section{Overview}
\label{sec:overview}

We begin with a running example and an overview of the problem and our solution approach.

\begin{example}
\label{ex}
Consider a robot in the environment as shown in Fig.~\ref{fig:ex} with
two other dynamic obstacles. The position of this robot is represented
by variables $x$ and $y$ in the coordinate system and the initial
position is at $x_0=4$ and $y_0= 3$. At each time instance, it can
apply the control input $u$ to change its position. The domain of $u$
is $Dom(u)=\Sigma= \{\sigma_1= (2,0)^{T}, \sigma_2= (-2,0)^{T},
\sigma_3= (0,1)^{T}, \sigma_4=(0,-1)^{T}\}$. At each time, with input
$\sigma_1$ (resp. $\sigma_2$) the robot can move in the $x$-direction
precisely with $2$ (resp. $-2$) units, however, in the $y$-direction
there is uncertainty: by $\sigma_3$ (resp. $\sigma_4$), the robot
proceeds some distance ranging from $1$ to $1.5$ (resp. from $-
1.5$ to $-1$) unit.  
There are two uncontrollable moving
obstacles, obj1 and obj2, whose behaviors are not known a priori but are known to
satisfy certain temporal logic formulas. 
Suppose as an
example design question that the available sensor for $y$ has slow
sampling rate, for example, the value of $y$ cannot be observed at
every time instance.
Can it eventually reach and stay in $R_2$ while
avoiding all the obstacles and not hitting the walls?
\end{example}
\begin{figure}[H]
\vspace{-2ex}
  \centering
  \includegraphics[width=0.3\textwidth]{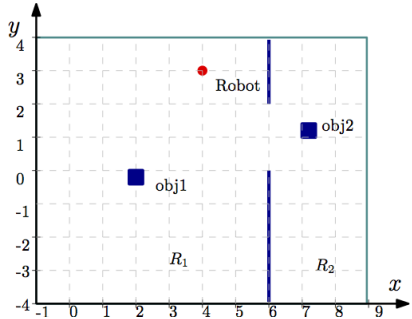}
  \caption{An environment including a robot (represented by the
    red dot) and two dynamic obstacles, obj1, obj2. Regions $R_1$ and
    $R_2$ are connected by a door.}
  \label{fig:ex}
\end{figure}
\vspace{-2ex}

A reactive controller senses the environment and decides an
action in response based on that
sensor reading (or a finite history of sensor readings).  For control
synthesis in reactive systems with partial observation, two problems
are critical.  One is a synthesis problem: given the \emph{current} sensor
design, is there a controller that realizes the specification?  Another
is a design problem: given an unrealizable specification, would it be
possible to find a controller by introducing new sensing modalities?
If so, what are the necessary modalities to add?

To answer these questions, we consider the counterexample guided
abstraction refinement procedure for two-player games with partial
observation in \cite{DimitrovaF08} . First, we formalize the interaction between a system
and its environment as a \emph{(concrete) game}.  A
safety specification determines the winning conditions for both
players. Then, an initial set of  predicates is selected to construct an
\emph{abstract game} with finite state space. The abstraction is
\emph{sound} in the sense that if the specification is realizable with the
system's partial observation in the abstract game, then it is so in
the concrete game.  However, if there does not exist such a
controller, a \emph{counterexample} that exhibits a violation of the
specification can be found. The procedure checks whether
this counterexample exists in the concrete game. If it does not, i.e., it is 
spurious, then the abstract game is refined until a
controller is obtained, or a genuine counterexample is found.

In the latter case, the task is not realizable by the system with its
current sensor design.  Then, we check whether it is realizable under
the assumption of complete information, using the same abstraction
refinement procedure. If the answer is yes, then the set of predicates
obtained in the abstraction refinement indicates the sensing
modalities that are sufficient, with respect to the given specification.

\section{Problem formulation}
\label{sec:problemformulation}

In this section we provide necessary background for presenting the
results in this paper.  
For a variable $x$ we denote with $Dom(x)$ its domain. Given a set of
variables $X$, a \emph{state} $v$ is a function $v: X \rightarrow
\bigcup_{x\in X} Dom(x)$ that maps each variable $x$ to a value in
$Dom(x)$. For $Y \subseteq X$, we write $v(Y)$ for the projection
of $v$ on $Y$.  Let the set of states over $X$ be $V$.  A
\emph{predicate} (atomic formula) $p$ is a statement over a set of variables $X$.  For
a given state $v$, $p$ has a unique value ---$\truev$ (1) or $\falsev$
 (0). We write $p(v)=1$ if $p$ is evaluated to $\truev$ by the state
 $v$. Otherwise, we write $p(v)=0$. Given a state $v \in V$, we write $v\models \varphi$, if the valuation of
$\varphi$ at $v$ is $\truev$. Otherwise, we write $v\not \models
\varphi$.
Given a formula $\varphi$ over a
set of predicates $\mathcal{P}$, let $\mathsf{Preds}(\varphi)\subseteq
\mathcal{P}$ be the set of predicates that occur in $\varphi$. 
 A \emph{substitution} of all variables $X$ in $\varphi$ with
the set of new variables $X'$ is denoted $\varphi(X')$.

\subsection{The model}
\label{subsec:model}
A \emph{(first-order) transition system} symbolically represents an
infinite-state transition system \cite{cimatti2012smt}.
\begin{definition}
A \ac{ts} $\mathbb{C}$ is a tuple $\langle
X, \mathcal{T}, \varphi_{init}\rangle $ with components as follows.
\begin{itemize}
\item $X$ is a finite set of variables. 

\item $\mathcal{T} (X,X')$ is a (quantifier-free) first-order logic formula describing the
  transition relation. $\mathcal{T}$ relates the variables
  $X$ which represent the \emph{current} state, with the
  variables $X'$ which represent the state after this
  transition.
 
\item $\varphi_{init}$ is a (quantifier-free) first-order formula over $X$
  which denotes the set of initial states  of
  $\mathbb{C}$. 
\end{itemize}
\end{definition}

The interaction between system and its environment is captured by a
reactive system formalized as a \ac{ts}.  
  \begin{example}
\label{ex2}
    We consider a modified version of Example \ref{ex} in which the
    environment does not contain any obstacle or internal walls.
    The set of variables is $X= \{x,y, u, t\}$ where $t$ is a
    Boolean variable. When $t=0$, the values of variables $x,y,u$ 
    are updated. Formally,  the transition relation is

$\arraycolsep=1.4pt
\begin{array}{rcl}
  \mathcal{T} & : = &  \big( t \land  t' = \neg t \land x'= x \land y'= y \land (\lor_{\sigma_i \in \Sigma} u'=\sigma_i ) \big)\\ 
              & \lor & \big( \neg t  \land  t' = \neg t  
              \land \big( ( u= \sigma_1 \land x'=x+2 \land y'=y ) \\&
&                              \lor    ( u= \sigma_2 \land x'=x-2 \land y'=y ) \lor \\&&
                                    ( u= \sigma_3 \land x'=x \land y' \ge y+1 \land y' \le y+1.5) \lor \\&&
                                   ( u= \sigma_4
                                \land x'=x \land y'\le y- 1 \land y'\ge y-1.5) \big)\big).  
\end{array}
$
Initially, $\varphi_{init}:= x=4 \land y= 3$ holds.
  \end{example}
  A \ac{ts} can be considered in a game formulation in which the
  system is player 1 and the environment is player 2. For this
  purpose, the set of variables $X$ is partitioned into $X_I \cup X_O
  \cup\{t\}$, where $X_I$ is the set of \emph{input variables},
  controlled by the environment, and $X_O$ is the set of
  \emph{output variables}, controlled by the system, and $t$ is a
  Boolean \emph{turn variable} indicating whose turn it is to make a transition:
  $1$ for the system and $0$ for the environment.  In
  Example~\ref{ex2}, the set of input variables is $X_I=\{x,y\} $, the
  set of output variables is $X_O=\{u \}$, and the turn variable is
  $t$.  We assume the domain of each output variable
  is finite. Without loss of generality \footnote{For a set of output
    variables, each of which has a finite domain, one can always
    construct a single new output variable to replace the set, and the
    domain of this new variable is the Cartesian product of the
    domains of these output variables.}, let $X_O$ be a
  \emph{singleton} $X_O=\{u\}$ and $Dom(u)= \Sigma$, which is a finite
  alphabet.

A  \ac{ts} $\mathbb{C}$ defines a game structure. In
this paper, we
assume that the system and its environment do not perform concurrent
actions, and thus the game structure is turn-based.
  \begin{definition} A \emph{game structure} capturing the
    interactions of a system (player 1) and its environment (player 2)
    in a \ac{ts} $\mathbb{C} = \langle
X, \mathcal{T}, \varphi_{init}\rangle$ is a tuple $G= \langle V, T, I\rangle$
    \begin{itemize}
\item $V=V_1\cup V_2$ is the set of states over $X$. $V_1  =\{v\in V\mid v(t)=1\} $ is the set of states at which player 1 makes a
  move ($t=1$). $ V_2= V\setminus V_1 $ consists of the states at
  which player 2 makes a move.
 \item $T =T_1\cup T_2$ is the transition relation: 
\begin{itemize}
 \item $((x_I,x_O,1), (x'_I,x'_O, 0) ) \in T_1$ if and only if $x_I=x_I'$
  and $\mathcal{T}((x_I,x_O,1), (x'_I,x'_O, 0) )$ evaluates to $\truev$.
\item $((x_I,x_O,0), (x'_I,x'_O, 1) ) \in T_2$ if and only if $x_O= x'_O$ and
  $\mathcal{T} ((x_I,x_O,0), (x'_I,x'_O, 1) )  $ evaluates to $\truev$.
\end{itemize}
 \item $I =\{v\in V \mid v \models \varphi_{init}\}$ is the set of
   initial states.
\end{itemize} 
\end{definition}
A \emph{run} is a finite (or infinite) sequence of states $\rho
=v_0 v_1 v_2\ldots \in V^\ast$ (or $\rho \in V^\omega$) such
that $(v_i,v_{i+1}) \in T$, for each $ 0 \le i < \abs{\rho}$
where $\abs{\rho}$ is the length of $\rho$. We assume the game is
nonblocking, that is, for all $v\in V$, there exists $v'\in V$ such
that $(v,v')\in T$. This can be achieved by including 
``idle'' action in the domain of the output variable.


\begin{definition}[Sensor model]
  Assuming the output variable $u$ and the Boolean variable $t$ are
  globally observable, the sensor model is given as a set of formulas
  $\{\mathcal{O}_x\mid x\in X_I\}$, where for each input $x\in X_I$,
  $\mathcal{O}_x$ is a formula over the set of input variables $X_I $
  such that the value of the input variable $x$ is observable at state
  $v$ if and only if the formula $\mathcal{O}_x $ evaluates to true at
  the state $v$.
\end{definition}

For a state $v\in V$, the set of \emph{observable variables} at $v$ is
$\mathsf{Obs}_{X}(v)=\{x\in X_I\mid v\models \mathcal{O}_x \} \cup
\{t,u\}$.  The \emph{observation} of $v$ is $\mathsf{Obs}(v)=
v(\mathsf{Obs}_X(v))$, which is the projection of $v$ onto the set of
variables observable at $v$.  Two states $v,v'$ are \emph{
  observation-equivalent}, denoted $v\equiv v'$ if and only if
$\mathsf{Obs}(v)=\mathsf{Obs}(v')$. The observation-equivalence can be
extended to sequences of states: let $\mathsf{Obs}(\epsilon) =
\epsilon$ and $\mathsf{Obs}(v
\rho)=\mathsf{Obs}(v)\mathsf{Obs}(\rho)$, for $v\in V$ and $\rho \in
V^\ast$(or $V^\omega$). Two runs $\rho,\rho' \in V^\ast$($V^\omega$) are observation equivalent, denoted $\rho\equiv \rho'$,
if and only if $\mathsf{Obs}(\rho)=\mathsf{Obs}(\rho')$.

This sensor model is able to capture both global and local sensing
modalities: if a variable $x$ is globally observable (globally
unobservable), $\mathcal{O}_x =\top $ (resp. $\mathcal{O}_x
=\perp$). Here $\top$ and $\perp$ are symbols for unconditional true
and false, respectively. As an example of a local sensing modality,
consider a sensor model in which an obstacle at $(px,py)$ is
observable if it is in close proximity of the robot at $(x,y)$, can be
described as $\mathcal{O}_{px} =(-2 \le px - x \le 2 )\land (-2 \le
py-y \le 2) \land \mathcal{O}_x\land \mathcal{O}_y$.

\subsection{Specification language}
We use \ac{ltl} formulas \cite{emerson1990temporal} to specify a set
of desired system properties such as safety, liveness, persistence  and
stability.  

In this paper, we consider safety objectives: the given specification
is in the form $\square \neg \varphi_{err}$, where $\square$ is the
\ac{ltl} operator for ``always'' and $\varphi_{err} $ is a formula
specifying a set of unsafe states $E=\{v \in V \mid v\models
\varphi_{err}\}.$ The objective of the system is to always avoid the
states in $E$ and the goal of the environment is to drive the game
into a state in $E$.

Let $v_0\in I$ be the designated initial state of the system.  We
obtain the game $\mathcal{G}^c=\langle V, v_0, T,E\rangle$,
corresponding to the reactive system $\mathbb{C}$ with the initial
state $v_0$.  From now on, $\mathcal{G}^c$ and $\mathbb{C}$ are
referred to as the \emph{concrete} game and \emph{concrete} reactive
system, respectively. The state set $V$ is the set of \emph{concrete}
states.  
A run $\rho \in V^\omega$ is winning for player 1 if it does not 
contain any state in the set of unsafe states $E$.

A \emph{strategy} for player $i$ is a function $f_i:V^\ast V_i
\rightarrow V_j$ which maps a finite run $\rho$ into a state
$f_i(\rho) \in V_j$, to be reached, such that $(v,v') \in T$, where $v$ is the last
state in $\rho$, $v' = f_i(\rho)$ and $(i,j)
\in\{(1,2), (2,1)\}$. The set of runs in $\mathcal{G}$ with the
initial state $v_0\in I$ induced by a pair of strategies $(f_1,f_2)$
is denoted by $Out_{v_0}(f_1,f_2)$.  Given the initial state $v_o$, a
strategy $f_1$ is winning for player 1, if and only if for any
strategy $f_2$ of player 2, any run in $Out_{v_0}(f_1,f_2)$ is winning
for player 1. A winning strategy for player 2 is defined dually.

Since the system (player 1) has partial observability, the strategies
it can use are limited to the following class.
\begin{definition}
  An \emph{observation-based} strategy for player 1 is a function
  $f_1: V^\ast V_1 \rightarrow V_2$ that satisfies: \begin{inparaenum}[(1)]
\item $f_1$ is a strategy of player 1; and 
\item for all $\rho_1,\rho_2$, if $\rho_1\equiv \rho_2$, then given
  $v= f_1(\rho_1), v'= f_1(\rho_2)$, it holds that for the output variable $u $, $v(u)=
  v'(u)$, and $v(t)=v'(t)$.
\end{inparaenum}
\end{definition}
For a game with partial observation, one can use knowledge-based subset
construction to obtain a game with complete observation. The winning
strategy for player 1 in the latter is an observation-based winning
strategy for player 1 in the former. The reader is referred to
\cite{chatterjee2006algorithm} for the solution of games with partial observation.
\subsection{Problem statement}
\noindent
We now formally state the problem investigated in this paper.
\begin{problem}
  Given a transition system $\mathbb{C}$ with the initial state $v_0
  \in I$, with a sensor model $\{\mathcal{O}_x \mid x\in X_I\}$ and
  a safety specification $\square \neg \varphi_{err}$, determine
  whether there exists an observation-based strategy (i.e. controller)
  $f_1$ such that for any strategy of the environment $f_2$ and for
  any $\rho \in Out_{v_0}(f_1,f_2)$, $\rho \models \square \neg
  \varphi_{err}$. If no such controller exists, then determine a new
  sensor model for which one can find such a controller, if there
  exists one.
\end{problem}

\section{Predicate abstraction}

Since the game $\mathcal{G}^c$ may have a large number of states, the synthesis methods for finite-state games cannot be directly applied or are not efficient. To remedy this problem, we apply an abstraction procedure which combines predicate abstraction and knowledge-based
subset construction and yields an
\emph{abstract finite-state} game with \emph{complete} information $\mathcal{G}^a$
from the (symbolically represented) concrete game $\mathcal{G}^c$.

\subsection{An abstract game}

Given a \emph{finite} set of predicates, the  abstraction procedure
constructs a finite-state reactive system (game structure).  
Let
$\mathcal{P}=\{p_1,p_2,\ldots,p_{N}\} $ be an indexed set of
predicates over variables $X$.  The \emph{abstraction function}
$\alpha _{\mathcal{P}}: V\rightarrow \{0,1\}^{\abs{\mathcal{P}}}$ maps
a concrete state into a binary vector as
follows. 
\[
\alpha _{\mathcal{P}} (v)= s \in \{0,1\}^{\abs{\mathcal{P}}} \text{ iff }\ s(i)=p_i(v),
\textrm{for all}\ p_i\in \mathcal{P},
\]
where $s(i)$ is the $i$th entry of binary vector $s$.  
The  \emph{concretization function} $\gamma_{\mathcal{P}}:
\{0,1\}^{\abs{\mathcal{P}}} \rightarrow 2^V $ does the reverse:
\[ \gamma_{\mathcal{P}}(s)= \{v\mid \forall\ p_i\in \mathcal{P}.\ 
p_i(v)= s(i)\}.
\] 

In the following, we omit the subscript $\mathcal{P}$ in the notation
for the abstraction and concretization functions wherever they are
clear from the context.  The following lemma shows that with a proper
choice of predicates, we can ensure that a set of concrete states
grouped by the abstraction function shares the same set of observable
and unobservable variables.
\begin{lemma}
\label{lm1}
Let $\bigcup_{x\in X_I} \mathsf{Preds} (\mathcal{O}_x)\subseteq
\mathcal{P}$. Then for any binary vector $s \in
\{0,1\}^{\abs{\mathcal{P}}}$ and any two states $v,v'\in \gamma(s) \ne
\emptyset$, it holds that $
\mathsf{Obs}_X(v) = \mathsf{Obs}_X(v')$.
\end{lemma}
\begin{proof}Since for any $v,v'\in \gamma( s)$, $\alpha(v)=\alpha(v')=s$, 
  for any $p\in \mathcal{P}$, $p$ has the same truth value at states $v
  $ and $v'$. Thus, for any $x\in X_I$, the formula
  $\mathcal{O}_x$, for which $\mathsf{Preds}(\mathcal{O}_x)\subseteq\mathcal{P}$, has the same
  value at $v$ and $v'$. Hence, if $x$ is observable (or
  unobservable) at $v$, then it must be observable (or unobservable)
  at $v'$ and vice versa. 
\end{proof}

Intuitively, by including the predicates in the formulas
defining the sensor model, for each $s \in \{0,1\}^{\abs{\mathcal{P}}}$, the set of 
concrete states $\gamma(s)$ share the same sets of observable and
unobservable variables. Hence, we use $X_v(s)$ to denote set of
observable/visible input variables in $s$ and $X_h(s) = X_I \setminus X_v(s)$ for the set of
unobservable/hidden input variables.

A predicate $p$ is observable at a state $v$ if and only if the variables
in $p$ are observable at $v$. According to Lemma~\ref{lm1}, if
there exists $ v\in \gamma(s)$ such that $p$ is observable at $v$, then $p$ is
observable for all $v\in \gamma(s)$ and we say that $p$ is observable at
$s$. Slightly abusing the notation $\mathsf{Obs}(\cdot)$, the
\emph{observation of a binary vector} $s$ is
$\mathsf{Obs}(s) = \{ (p_i,s(i))\mid p_i \text{ is observable at } s  \}$,
which is a set of assignments for observable predicates. Two binary
vectors $s,s'$ are observation-equivalent, denoted $s\equiv s'$, if and
only if
$\mathsf{Obs}(s)= \mathsf{Obs}(s')$.

The abstraction of the concrete game $\mathcal{G}^c = \langle
V,v_0, T, E\rangle$
with respect to a finite set of predicates $\mathcal{P}$ is a game
with \emph{complete information}
$\alpha(\mathcal{G}^c,
\mathcal{P}) =\mathcal{G}^a = \langle S^a, s_0^a, T^a, E^a\rangle$:
\begin{itemize}
\item $S^a= S_1^a \cup S_2^a$ is the set of abstract states with sets of player 1's and player 2's abstract states respectively, \\
 $S_1^a = \{ s^a \mid \exists v\in V_1.\ s^a \subseteq \{s\mid s\equiv \alpha(v)\},s^a\not = \emptyset\}$ and\\
 $S_2^a= \{ s^a \mid \exists v\in V_2.\  s^a \subseteq \{s\mid s\equiv \alpha(v)\},s^a\not = \emptyset\}$.
\item $s_0^a= \{s \in \{0,1\}^{\abs{\mathcal{P}}} \mid s \equiv
  \alpha(v_0)\}$ is the initial state.
\item $T^a = T_1^a\cup T_2^a$ where 
\begin{itemize}
\item $(s^a_1,s^a_2)\in T_1^a$ if and only if the following conditions
  (1), (3) and (4) are satisfied.
\item $(s_1^a, s_2^a)\in T_2^a$ if and only if the following
  conditions (2), (3) and (4)
  are satisfied.
\end{itemize}
\begin{enumerate}[(1)]
\item \label{must} for every
    $s \in s^a_1$ and every
  $v\in \gamma(s)$, there exist $s'\in s^a_2$ and
  $v'\in \gamma(s')$ such that $(v,v')\in T_1$; 
\item \label{may}  there exists $s\in s_1^a$, $v\in \gamma(s)$,
  $s'\in s_2^a$ and $v'\in \gamma(s')$ such that $(v,v')\in T_2$;
\item\label{rel2} for every  $s' \in s^a_2 $, there exist $s \in s^a_1$, $v\in \gamma(s)$ and $v'\in \gamma(s')$ such that
  $(v,v')\in T$;
\item \label{rel3} for every $s_1', s_2' \in \alpha(V)$, if
  $s_1'\in s^a_2$, $s_1' \equiv s_2'$ and there exist $s\in s_1^a$, $v\in \gamma(s)$ and $v'\in \gamma(s_2')$ with $(v,v')\in T$, then also $s_2 ' \in s^a_2$.
\end{enumerate}
\item $E^a = \{s^a \mid \exists s \in s^a.\ \exists v\in \gamma(s).\ v
  \in E\}$ is the set of unsafe states. 
\end{itemize}
In what follows, we refer to a state $s^a\in S^a$ as an \emph{abstract
  state}.  By definition, each $s^a$ in $\mathcal{G}^a$ is a set
of observation-equivalent binary vectors in $\alpha(V)$.


We relate a binary vector $s \in \{0,1\}^{\abs{\mathcal{P}}}$ with a formula $[s]$ that is a conjunction such that 
$[s] = \land_{0\le i \le \abs{\mathcal{P}}} h_i $ where if $s(i)=1$,
then $h_i = p_i$, otherwise $h_i = \neg p_i$. 
Further, for any $ s^a \in S^a$, we define the following formula in disjunctive normal
form $[s^a] =\lor_{s\in s^a} [s]$.

\addtocounter{example}{-1}
\begin{example}[cont.]We assume $x$ is globally
  observable and $y$ is globally unobservable and require
  that the robot shall never hit the boundary, that is, $\square \neg
  \varphi_{err}$ where $\varphi_{err} = \big(t =0\land (x\ge 9 \lor y
  \ge 4 \lor x \le -1 \lor y \le -4 ) \big)$. Let $\varphi_{init}: = (x= 4 \land y= 3 \land u = \sigma_1 \land t = 1)$.
 Let $\mathcal{P}= \{
  x\ge 9, y \ge 4, x \le -1, y\le -4, u=\sigma_1,
  u=\sigma_2,u=\sigma_3,u=\sigma_4, t=1\}$. The initial state of
  $\mathbb{C}$ is $v_0 =(4, 3, \sigma_1, 1)$, and the corresponding
  initial state in $\mathcal{G}^a$ is $s^a_0 =\{ (00001 000 1) \}$
  where the values for the predicates in $s^a_0$ are given in the same
  order in which they are listed in $\mathcal{P}$. Given $v'= (4, 3 ,
  \sigma_2, 0)$, since $(v_0,v')\in T$, we determine $(s^a_0 ,s^a_1)
  \in T^a$ where $s^a_1= \{ (000 0 01000)\} $ indicating $u=\sigma_2$
  and $t=0$.
\end{example}
We show that by a choice of predicates, it is ensured that for any
$s^a \in S^a$, all concrete states in the set $ \{v\mid \exists s\in
s^a.\ v\in \gamma(s)\}$ share the same observable and unobservable
variables.
\begin{lemma}
  If $\bigcup_{x\in X_I} \mathsf{Preds} (\mathcal{O}_x)\subseteq
  \mathcal{P}$, then for any $s^a\in S^a$ and $v,v'\in \{v\mid
  \exists s\in s^a.\ v\in \gamma(s)\}$, it holds that
  $\mathsf{Obs}_X(v)=\mathsf{Obs}_X(v')$.
\end{lemma}
\begin{proof}
  By Lemma \ref{lm1}, since for any $s\in
  \{0,1\}^{\abs{\mathcal{P}}}$, for any $v,v' \in \gamma(s) \ne
  \emptyset$, $\mathsf{Obs}_X(v)=\mathsf{Obs}_X(v')$, then it suffices to
  prove that for any $s,s'\in s^a$, $X_v(s)=X_v(s')$ and
  $X_h(s)=X_h(s')$. By definition, $s \equiv s'$ implies that the set
  of observable (unobservable) predicates is the same in both $s$ and
  $s'$. Thus, the set of observable (unobservable) variables that
  determines the observability of predicates has to be the same in both $s$ and $s'$. That is, $X_v(s)=X_v(s')$ and $X_h(s)=X_h(s')$.
\end{proof}Let $X_v(s^a)$ (resp. $X_h(s^a)$) be the
observable (resp. unobservable) input variables in the abstract  state
$s^a$. That is, $X_i(s^a)=X_i(s)$ for any $s\in s^a$, for $i\in \{v,
h\}$.

\subsection{Concretization of strategies}
In the abstract game $\mathcal{G}^a$, there exists a winning strategy
for one of the players. We show that a winning strategy for the system
in $\mathcal{G}^a$ can be concretized into a set of observation-based
winning strategies for the system in $\mathcal{G}^c$.

For $(i,j)\in\{(1,2),(2,1)\}$, the
  \emph{concretization of a strategy} $f_i: (S^a)^\ast S^a_i
  \rightarrow S^a_j$ in $\mathcal{G}^a$ is a \emph{set} of strategies in $\mathcal{G}^c$, denoted
  $\gamma(f_i)$ and can be obtained as follows. Consider $\rho^c \in
  V^\ast$, $\rho \in S^\ast$, $\rho^a\in (S^a)^\ast$ in the following,
  where
\vspace{-1ex}
\[\arraycolsep=2pt
\begin{array}{l ccccccl}
\rho^c &= & v_0 & v_1 &v_2 & \ldots & v_n &,\\
\rho &= & s_0 & s_1 & s_2 & \ldots & s_n  &, \\
\rho^a &= & s^a_0 & s^a_1 & s^a_2 & \ldots & s^a_n &. \\
\end{array}
\]  
and $v_i \in \gamma(s_i)$, $s_i\in s^a_i$ for each $i: 0\le i\le n$.
Given $f_i(\rho^a ) = s^a_{n+1}$, the output $f_i^c (\rho^c)= v_{n+1}
$ such that there exist $ s\in s^a_{n+1}$ and $v_{n+1}\in \gamma(s) $ such that $
(v_n,v_{n+1})\in T$. In other words, $v_{n+1}$ is a concrete state
reachable from the current state $v_n$ and can be abstracted into a
binary vector $s$ in the abstract state $s^a_{n+1}$. Intuitively,
given the run $\rho^c$, one can find a run in the abstract system
$\rho^a$, and uses the output of $f_i$ on $\rho^a$ to generate an
abstract state. Then $f_i^c$ picks a reachable concrete state, which
can also be abstracted into a binary vector contained this abstract
state. A strategy $f$ is \emph{concretizable} if $\gamma(f)\ne
\emptyset$. Otherwise it is \emph{spurious}.

\begin{theorem}\label{thm:soundness}
The concretization $\gamma(f_1)$ of a player 1's winning strategy
$f_1: (S^a)^\ast S^a_1\rightarrow S^a_2$ in $\mathcal{G}^a$ is a non-empty set that consists of observation-based winning
strategies for player 1 in the concrete game $\mathcal{G}^c$.
\end{theorem}
\begin{proof}
Follows from the proof in \cite{rayna-report}.
\end{proof}

In case there is no winning strategy for player 1 in $\mathcal{G}^a$,
the synthesis algorithm gives us a winning strategy for player 2 in $\mathcal{G}^a$,
which we refer to as \emph{counterexample}. Then we need to check if
it is spurious, as explained in the next
section.

\section{Abstraction refinement}
We consider an initial set of predicates $\mathcal{P}$ which consists
of the predicates occurring in $\varphi_{err}$, the predicates
describing the output $u$ of the system, and those
occurring in the sensor model.  With this initial choice of predicates,
if player 1 wins the game $\mathcal{G}^a =
\alpha(\mathcal{G}^c,\mathcal{P})$, then the abstraction does not need
to be further refined,  according to Theorem \ref{thm:soundness},
the winning strategy of player 1 is concretizable in the concrete
game.  However, if player 2 wins, there exists a deterministic winning strategy $f_2: (S^a)^\ast
S^a_2 \rightarrow S^a_1$ in the game
$\mathcal{G}^a$. The next step is to check if $f_2$ is
spurious. If it is, then the abstract model is too coarse and needs to
be further refined.

\subsection{Constructing abstract counterexample tree}
We construct a formula from the strategy tree generated from this
counterexample that characterizes the concretizability of this counterexample in the
concrete system $\mathbb{C}$, and then we construct a formula from the
 tree. If the formula is satisfiable, then the counterexample is genuine.

 Given the initial state $s^a_0 $, the \ac{act} for $f_2$ is
 $\mathbb{T}(f_2,s^a_0 ) = (\mathcal{N}, \mathcal{E})$ where
 $\mathcal{N}$ are nodes and $\mathcal{E} \subseteq \mathcal{N}\times
 \mathcal{N}$ are edges.  Each node $n$ in $\mathcal{N}$ is labeled by
 a state $s^a\in S^a$ and we denote the labeling $n:s^a$.  A node
 $n:s^a$ belongs to player $i$ if $s^a \in S^a_i$, for $i =1,2$.

In the case of a safety specification, $\mathbb{T}(f_2, s^a_0 )$ is a
finite tree in which the following conditions hold. 
\begin{inparaenum}[1)]
\item The root $0$ is labeled by $s^a_0$, that is, $0: s^a_0$.
\item If $n:s^a$ is a player 1's node and $n$ is not a leaf, then for
  each $t^a $ such that $(s^a,t^a)\in T^a$, add a new child $m$ of $n$
  and label $m$ with $t^a$. Let $n \xrightarrow{\sigma} m$ for which
  $[t^a] \implies u=\sigma $.
\item If $n: s^a$ is a player 2's node and $n $ is not a leaf, then
  add one child $m$ of $n$, labeled with $t^a= f_2(\rho)$, where
  $\rho$ is the sequence of nodes' labels (states) on the path from
  the root to the node $n$. Let $n \xrightarrow{\epsilon}m$ where
  $\epsilon$ is the empty string. 
\item For a node $n:s^a$, $n$ is a leaf if either $s^a\in E^a$ or there is no outgoing transition from $s^a$.
\item Each node has at most one parent.
\end{inparaenum}


We illustrate the \ac{act} construction on the small
example. Fig.~\ref{ex:act} shows a fragment of \ac{act} for
Example~\ref{ex2}.  First we define the root $0$, labeled with the
abstract state $s^a_0$. At $s^a_0$,
player 1 can select any output in $\Sigma$.  Therefore, the children
of $0$ are $1,2,3,4$, one for each input in $\Sigma$.  For instance,
the output $\sigma_2$ labels the edge from $0$ to $2$ and we have
$2:s^a_2$. The only child of $2$ is $6$, labeled with $s^a_6 =
\{(001001001)\}$.  Clearly, the actual value of $x$ after executing
$\sigma_2$ is $2$. Yet the reached state $s_6^a$ in which the
predicate $x \le-1$ is $\truev$ is because there exists some $x \in
(-1, 1] $ at state $s_2^a$, and will make $x \le -1$ satisfied after
action $\sigma_2$. This is caused by the coarseness of the
abstraction.  If player 1 takes action $\sigma_3$, then it will have
no information about the value of the predicate $(y\ge 4)$, as this
predicate is not observable. In Fig.~\ref{ex:act}, each state $s^a_i,
0\le i \le 7$ is related with a formula $[s^a_i]$ (shown below the figure).
\vspace{-2ex}
\begin{figure}[H]
\centering
\includegraphics[width=0.45\textwidth]{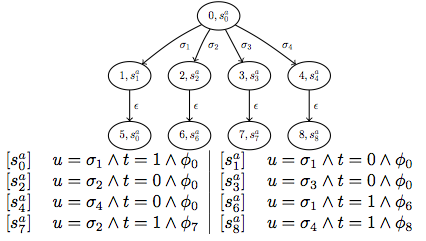}
\caption{A fragment of \ac{act} for Example~\ref{ex2}. }
\label{ex:act}
\end{figure}
\vspace{-2ex}
Note that nodes
  $5$ and $0$ are labeled with the same state $s^a_0$. In the formulas
  $[s^a_i]$, 
$\begin{array}{lll}
  \phi_0 & = & \neg (x\ge 9) \land \neg (x \le-1) \land \neg (y\ge 4) \land \neg (y \le-4), \\
  \phi_6 & = & \neg (x\ge 9) \land (x \le-1) \land \neg (y\ge 4) \land \neg (y \le-4), \\
  \phi_7 & = & \neg (x\ge 9) \land \neg (x \le-1) \land \neg (y \le-4),  \\
  \phi_8 & = & \neg (x\ge 9) \land \neg (x \le-1) \land \neg (y\ge 4). \\
 \end{array}
 $

For a node $n\in \mathcal{N}$, $C(n)$ is the set of children of $n$
and $\mathsf{Paths}(n) \subseteq \mathcal N^*$ is the set of paths
from the node $n$ to a leaf.  For a path $\rho \in \mathcal{N}^\ast$,
the \emph{trace} of $\rho$, denoted $\mathsf{Trace}(\rho)\in
\Sigma^\ast$, is the sequence of labels on the edges in the path.  A
node $n \in \mathcal{N}$ is related with a set of traces
$\mathsf{Traces}(n)= \{ \mathsf{Trace}(\rho)\mid \rho \in
\mathsf{Paths}(n)\}$. For a leaf, $\mathsf{Traces}(n)=\{\epsilon\}$ by
default. For example, $\mathsf{Trace}(0\xrightarrow{\sigma_1}1
\xrightarrow{\epsilon} 5)=\sigma_1 \epsilon= \sigma_1$.

Note that in the tree structure defined here, for each node $n\in \cal
N$, there exists exactly one path from the root to $n$, and hence
there is one trace $w \in \Sigma^\ast$ that labels that path.



We annotate each node $n:s^a$ with a set of variables $\mathcal{X}^n$
as $n:s^a:\mathcal{X}^n$ where $\mathcal{X}^n=\{X^{n,w} \mid w\in
\mathsf{Traces}(n) \}$ and $X^{n,w} = (X_v^{n}, X_h^{n,w},
u^{n},t^{n})$ where $X_v^{n} \cup \{t^{n},u^{n} \}$ are observable
variables in $s^a$ and $X_h^{n,w}$ are hidden variables in $s^a$ when
the trace from $n$ is $w$.  
For example, we annotate $0$ with
$\mathcal{X}^0=\{X^{0,w} = (x^0, y^{0,w}, u^0, t^0)\mid
w\in\mathsf{Traces}(0) =\{ \sigma_1,\sigma_2,\sigma_3,\sigma_4\}\}$ as
$y$ is not observable.  With this annotation, the unobservable
variables $X_h$ at node $n$ can be assigned with different values for
different traces from $n$. It corresponds to the fact that the
concrete states, grouped into an abstract state, share the same values
for observable variables but may have different values for
unobservable ones.

In what follows, we relate a trace with a \emph{trace formula}. By
checking the satisfiability of a \emph{tree formula}, built from trace
formulas with a \ac{smt} solver, we can determine whether the
counterexample is spurious.

\subsection{Analyzing the counterexample}
Given a trace $w \in \mathsf{Traces}(n)$, the trace formula $F(n,w)$ is
constructed recursively as follows. 

\begin{itemize}
\item If $n:s^a : \mathcal{X}^n$ is a leaf, then $\mathcal{X}^n
  =\{X^{n,\epsilon}\}$ is a singleton. Let $
  F(n,w)=[s^a](X^{n,\epsilon}) $, which is satisfiable if there exists
  a concrete state $v$ for $X^{n,\epsilon}$ such that
  $[s^a](v)=\truev$.
  
\item If $n:s^a:\mathcal{X}^n$ is a player 1's node and not a leaf,
  then for each $w =\sigma w' \in \mathsf{Traces}(n)$, for each child
  $m:t^a:\mathcal{X}^m$ such that $n\xrightarrow{\sigma} m$, let  
\vspace{-1ex}
  \begin{multline*} F(n,m,w)=F(m,w') \land [s^a](X^{n,w})\\ \land
    [t^a](X^{m,w'}) \land u^m =\sigma \land \mathcal{T}(X^{n,w},X^{m,
      w'}) \end{multline*} where $F(m,w') $ is false if $w'\notin
  \mathsf{Traces}(m)$. Then let $F(n,w )= \lor_{m \in C(n),
    u^m=\sigma} F(n,m,w)$. Intuitively, $F(n,m,w)$ can be satisfied if
  there exist a state $v$ for $X^{n,w}$ and $v'$ for $X^{m,w'}$ such
  that $[s^a]$ and $[t^a]$ evaluate to $\truev$ at $v$ and $v'$,
  respectively; action $\sigma$ enables the transition from $v $ to
  $v'$; and $F(m,w')$ is satisfied. The disjunction is needed because
  for a node $n$, there can be more than one $\sigma$-successors.
 
\item If $n:s^a: \mathcal{X}^n$ is a player 2's node and not a leaf, there
  exists exactly one child of $n$, say, $m: t^a: \mathcal{X}^m$,
then for each $w\in \mathsf{Traces}(n)$, let $
  F(n,w)=F(m,w)\land [s^a](X^{n,w})\land [t^a](X^{m,w}) \land
  \mathcal{T}(X^{n,w},X^{m,w}).$
\end{itemize}
The \emph{tree formula} is 
\vspace{-1ex}
 \[F(0)= \land_{w\in \mathsf{Traces}(0)} (F(0,w) \wedge
\varphi_{init}(X^{0,w})).\]
\begin{theorem}
  Let $f_2$ be a winning strategy for the environment in the game
  $\mathcal{G}^a$, the strategy $f_2$ is genuine, i.e.,
  $\gamma(f_2)\ne \emptyset$, if and only if the tree formula $F(0)$ is
  satisfiable.
\end{theorem}
\begin{proof}
The reader is referred to  \cite{DimitrovaF08}.
\end{proof}
\addtocounter{example}{-1}
\begin{example}[cont.]
 Consider, for instance, the trace
  $\sigma_1w'\in \mathsf{Traces}(0)$ corresponds to a labeled path
  $ 0\xrightarrow{\sigma_1} 1 \xrightarrow{\epsilon } 5 \text{ and } w'\in
  \mathsf{Traces}(5)$. Since $0\xrightarrow{\sigma_1}1$, we have 
$F(0, \sigma_1w')= F(1,w') \land [s^a_0](X^{0, \sigma_1w'})\land
[s^a_1](X^{1,w'}) \land 
u^1=\sigma_1\land \mathcal{T}(X^{0, \sigma_1w'}, X^{1,w'})
$ where $X^{1,w'}= (x^1,y^{1,w'},u^1,t^1)$. Then given $1\xrightarrow{\epsilon}5$,
$F(1,w')= F(5,w') \land [s^a_1](X^{1,w'}) \land [s^a_5](X^{5,w'})
 \land
\mathcal{T}(X^{1,w'}, X^{5,w'})
$, where $X^{5,w'}= (x^5,y^{5,w'},u^5,t^5)$. In above equations, for instance
$[s^a_0](X^{0,\sigma w'})= \neg (x^0 \ge 9 \lor x^0 \le -1 \lor
y^{0,\sigma w'} \le -4 \lor y^{0,\sigma w'} \ge 4) \land
 u^0=\sigma_1 \land t^0=1$.
\end{example}

\subsection{Refining the abstract transition relations}
\label{subsec:refinetrans}

Given a node $n$ and a trace $w \in \mathsf{Traces}(n)$, if $F(n,w)$
is unsatisfiable, then the occurrence of the spurious counterexample
is due to the approximation made in abstracting the transition
relation. To rule out this counterexample, we need to refine the
abstract transition relation. For this purpose, we define a \emph{node
  formula} $\tilde{F}(n,w)$ as described below.

First, we define the \emph{pre-condition of a formula}: for a
formula $\varphi$ and $\sigma\in \Sigma$, the pre-condition of
$\varphi$ with respect to $\sigma$, $\textsc{Pre}_1(\sigma,\varphi)$
is a formula such that $v \models \textsc{Pre}_1(\sigma, \varphi)$ if and only if
there exists $v'\in V$ such that $v' \models \varphi$, $v'(u)=\sigma$
and $ (v,v') \in T_1 $. Intuitively, at any state $v$ that satisfies
this formula $\textsc{Pre}_1(\sigma,\varphi)$, the system, after
initiating the output $\sigma$, can reach a state $v'$ at which
$\varphi$ is satisfied.  Let $\textsc{Pre}_1(\varphi)=\lor_{\sigma \in
  \Sigma} \textsc{Pre}_1(\sigma,\varphi)$. Correspondingly,
$\textsc{Pre}_2(\varphi)$ is a formula such that $v \models
\textsc{Pre}_2(\varphi)$ if and only if there exists $v' \in V$, $v'\models \varphi$ and
$(v,v')\in T_2$.

Now, we define the \emph{node formula} $\tilde{F}(n,w)$ as follows.

\begin{itemize}
\item If $n:s^a$ is a leaf node, then $w=\epsilon$
  and $\tilde{F}(n,\epsilon)=\lor_{s\in s^a, [s]\implies \varphi_{err}} [s] $.
\item If $n:s^a$ belongs to player 1 and is not a leaf, and $w =\sigma
  w'$, then
\vspace{-1ex}
\[
\tilde{F}(n,w)= [s^a]\land \textsc{Pre}_1(\sigma, \lor_{\ell \in C(n),
  u^\ell = \sigma}
\tilde{F}(\ell,w')),
\]
where $\tilde{F}(\ell,w')$ is false if $w'\notin \mathsf{Traces}(\ell)$. Here, the set
$\{\ell \in C(n)\mid u^{\ell}=\sigma \}$ is a set of
$\sigma$-successors of $n$. 
\item If $n:s^a$ belongs to player 2's and is not a leaf, then 
\vspace{-0.5ex}
\[
\tilde{F}(n,w)= [s^a]\land \textsc{Pre}_2\left(\tilde{F}(m, w)\right)
\]
where $m\in C(n)$ is the unique child of node $n$.
\end{itemize}

We augment the current set $\mathcal{P}$ with all
predicates that occur in the formula $\tilde F(n,w)$, i.e.,
$
\mathcal{P}':=\mathcal{P}\cup \mathsf{Preds}(\tilde F(n,w)).
$
For each node $n$ and each $w\in \mathsf{Traces}(n)$ such that $
F(n,w)$ is unsatisfiable, the procedure generates a set of predicates
$\mathsf{Preds}(\tilde F(n,w))$, which are then combined with the
current predicate set to generate a new abstract game. We repeat
this procedure iteratively until a set of predicates is found such
that for any $n$ and any $w \in \mathsf{Traces}(n)$, $ F(n,w)$ is satisfiable.

\subsection{Refining the abstract observation equivalence}
\label{subsec:refineobeq}
If each trace formula for the considered counterexample tree is satisfiable, but the tree formula is not, then we need to check
whether the existence of a counterexample is because of the
coarseness in the abstraction observation-equivalence.

We are in the case when for all $w\in
\mathsf{Traces}(0)$,  $F(0,w) \wedge \varphi_{init}(X^{0,w})$ is satisfiable.  Let
$\Phi=\{F(0, w)\wedge \varphi_{init}(X^{0,w})\mid w \in \mathsf{Traces}(0)\}$. Since $F(0) = \land_{\phi\in \Phi} \phi$ is unsatisfiable, there exists a subset $\Psi$ of
$\Phi$ such that $\psi= \land_{\phi\in \Psi} \phi$ is satisfiable and
a formula $\varphi \in \Phi \setminus \Psi$ such that $\varphi\land \psi$ is unsatisfiable. Let the sets of
free variables in $\psi$ and $\varphi$ be $Y$ and $Z$
respectively. Since only observable variables are shared between
different traces, $Y\cap Z$ only consists of \emph{observable}
variables.

A \emph{Craig interpolant} \cite{smullyan1995first} for the pair $(\psi(Y), \varphi(Z) )$ is
a formula $\theta(Y\cap Z)$ such that \begin{inparaenum}
\item $\psi(Y) $ implies $ \theta(Y\cap Z)$,
\item $\varphi(Z) \land\theta (Y \cap Z)$ is unsatisfiable.
\end{inparaenum}
To illustrate, consider the following example. Let $\varphi_1=
(y^5=y^0+1)\land (y^5\ge 4)$ and $ \varphi_2=(y^0\le 1)$. Clearly,
$\varphi_1\land \varphi_2 \equiv \perp$ because $y^0$ in $\varphi_1$ needs
to satisfy $y^0\ge 3$.  Then the formula $\theta = y^0\ge 3$ is an
interpolant for the pair of formulas
$(\varphi_1(y^0,y^5),\varphi_2(y^0))$. For a number of logical
theories commonly used in verification, including linear real
arithmetic, Craig interpolants can be automatically computed
\cite{McMillan2011}.

After computing the interpolant $\theta$ for $(\psi,\varphi)$, we update the set of
predicates to be
$
\mathcal{P}':= \mathcal{P}\cup \mathsf{Preds}(\theta)
$. 
In the end, Algorithm~1 describes
the refinement procedure.
\vspace{-2ex}
\begin{figure}[ht]
\centering
\includegraphics[width=0.5\textwidth]{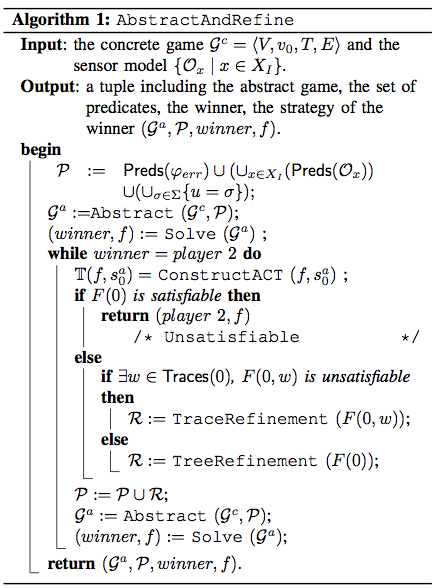}
\label{alg:abstractrefine}
\end{figure}
\vspace{-4ex}

\section{Sensor reconfiguration}
\label{subsec:sensorconf}
Suppose the task specification is unrealizable given the current
sensor model. Then, a prelude to refining the sensor is identifying
whether the source of unrealizability is limited sensing. To this end,
we first check whether it is realizable under the assumption that the
system has perfect observation over its environment. For this purpose,
we run the procedure \texttt{AbstractAndRefine} with the concrete game
$\mathcal{G}^c$ and a sensor model defined as $\{\mathcal{O}_x =\top
\mid x\in X_I\}$, which means all the input variables are globally
observable.  If player 1 wins the abstract game, then we can conclude
that the task is not realizable because of the limited sensing
capability.

The procedure \texttt{SensorReconfigure}, shown as Algorithm~2, computes
a set of predicates that we need to observe in order to
satisfy a given specification. The algorithm takes the concrete
system, its current sensor model and an unrealizable specification as
input. Then by making all variables observable, we use the procedure \texttt{AbstractAndRefine} to determine if the task is realizable given complete observation. If \texttt{AbstractAndRefine} terminates with a positive answer, then, the set of predicates obtained by the refinement suffices for realizing the specification. Further, the predicates involving
unobservable variables indicate the set of new sensing modalities to
be added, and provide the requirements on the sensors'
precision and accuracy for both observable and unobservable variables.

\vspace{-2ex}
\begin{figure}[ht]
\centering
\includegraphics[width=0.5\textwidth]{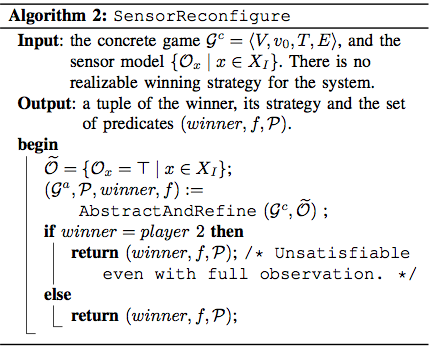}
\label{alg:sensor_reconfig}
\end{figure}
\vspace{-4ex}

\section{Case study}
\label{subsec:case}
We demonstrate the method by revisiting Example \ref{ex}. Assuming the
dynamics of obstacle obj1 with position $(x_p,y_p)$ is given in form
of logical formula $\varphi_{p}: = ((x \le 6 \land x_p' \ge 7) \lor (x
> 6 \land x_p' \ge 6)) \land \neg \varphi_{hit}$ where $\varphi_{hit}$
is a formula that is satisfied when the obstacle hits the wall or the
robot.  For obstacle obj2 $(x_o,y_o)$, we have $\varphi_{o}: = ((x \le
6 \land x_o '< 4) \lor( x > 6 \land x_o'\le 7)) \land \neg
\varphi_{hit}$. Here we have a liveness condition which specifies that
the robot has to visit and then stay within region $R_2$. To enforce
such constraint, we introduce a Boolean variable $err$ and set $err
=1$ if $x \ge 6\land x' < 6 $, which means if the robot in $R_2$
returns to $R_1$, an error occurs and the system reaches an unsafe
state.
\paragraph*{Case 1}
Due to the limited sampling rate in the sensor for variable $y$, the
system receives the exact value of $y$ intermittently (every other
step). In this case, we introduce a predicate $p_s$ such that if  $p_s =1$
then the exact value of $y$ is observed, otherwise there is no data
sampled. The transition relation $\mathcal{T}(X,X')$ is modified to
capture this type of partial observation. For example, given $u
=\sigma_3$, the transition is $ t \land( u=\sigma_3 )\land y' \ge y+1
\land y'\le y+1.5 \land x'=x \land \varphi_p \land ((\neg p_s \land
y_s' = y +1 \land y_n' = y+1.5 ) \lor ( p_s \land y_s' =y' \land y_n'
= y')) \land p_s' =\neg p_s \land t'= \neg t $ where $y_s$ and $y_n$
are auxiliary variables used by the robot to keep track of the upper
and lower bounds, respectively, for the value of $y$. Intuitively,
when there is no data received, the robot makes a move such that for
every $y$ within the upper and lower bounds, for all possible changes
in its obstacles, it will not encounter any unsafe state.  Then the
sensor data received in the next step resolves the ambiguity it
had earlier about $y$ and an action is selected accordingly.

The abstraction refinement procedure starts with an initial set of
$11$ predicates. After $17$ iterations, we obtained an abstract game
in which the system has a winning strategy. The abstract game is
computed from $45$ predicates and has $2516$ states. The computation
takes $5.8$ min in a computer with $4$ GB RAM, Intel Xeon
processors. The obtained predicates relating to the variable $y$ falls
into the following categories:
\begin{inparaenum}[(1)]\item
Predicates over the unobservable variable $y$: $y \le -4$, $ y \ge
4$, $y\le -2.5$, $y\ge 3.5$, $y \le -1$, $y\ge 2.5$, $y< 1.5$, $y \le
1.5$, $y \le 0$, $y\ge 2$, $y \le y_o$, $y < y_s$.
\item
Predicates over the observable variable $y_s$:  $y_s\le
1.5$, $y_s \ge 1.5$. And \item there is no predicate over the upper bound
variable $y_n$.
\end{inparaenum} The predicates relating to the obstacles $(x_p,y_p),
(x_o,y_o)$ are the following: $x_p \le x $, $x \le x_o$, $x_p
\le 2$, $x_p \le 7$, $x_p \le 6$, $x_o \le x $, $x_o < 4$, $x_p \le
6$, $x_o > 4$, $y\le y_o $. 

With the obtained set  $\mathcal{P}$ of predicates, we can decide the
requirement on the precision of sensor for this task. For
every $p\in \mathcal{P}$ , the constants in $p$ has at most one
decimal place, for example, $y\ge 2.5$. Thus, a sensor which can
reliably measure just one decimal place would suffice. Besides, there is no need to keep track of the upper bound $y_n$ for
$y$ and also the value of  $y_p$ for obj2.

\paragraph*{Case 2}In this case we consider the sensor model with an
extra limitation: the robot cannot observe obj2 if it is in $R_1$, or
obj1 when it is in $R_2$. To capture this local sensing modality, we
made $x_o,y_o,x_p,y_p$ unobservable and introduce another four
auxiliary observable variables $x_p^c,y_p^c$ and $x_o^c, y_o^c$. When
the robot is in $R_1$, the values of $x_o^c,y_o^c$ equal that of $x_o$
and $y_o$. But when it is in $R_2$, $(x_o^c,y_o^c)$ can be any point
in $R_1$ following the dynamic in robot's assumption of obj1.  Similar
rules applied to $x_p^c $ and $y^c_p$. For the same task
specification, after 21 iterations, which takes about $30$ min, the
abstraction refinement outputs an abstract system with $8855$ states
using $60$ predicates and finds the robot a winning strategy.

\section{Conclusion and future work}
\label{sec:conclusion}

We took a first step toward explicitly accounting for the effects of
sensing limitations in reactive protocol synthesis. The formalism we
put forward is based on partial-information, turn-based,
temporal-logic games. Using witnesses for unrealizability in such
synthesis problems and interpolation methods, we proposed an
abstraction refinement procedure. An interpretation of this procedure
is systematical identification of new sensing modalities and precision
in existing sensors to be included in order to construct feasible
control policies in reactive synthesis problems.  A potential
bottleneck of the proposed formalism is the rapid increase in the
problem size due to, for example knowledge-based subset
construction. A pragmatic future direction is to consider so-called
lazy abstraction methods \cite{henzinger2002lazy} for partial
observation control synthesis, so that different parts of the concrete
game can be abstracted using different sets of predicates. In this
manner, the system is abstracted with different degree of precision
and thus its sensor model can also be configured ``locally" for
different parts of the system.  Furthermore, besides precision, one
would also be interested in refinements in sensing with respect to
accuracy; therefore, extensions to partially observable stochastic
two-player games are also of interest.

\bibliography{mybib}

\begin{thebibliography}{10}
\providecommand{\url}[1]{#1}
\csname url@samestyle\endcsname
\providecommand{\newblock}{\relax}
\providecommand{\bibinfo}[2]{#2}
\providecommand{\BIBentrySTDinterwordspacing}{\spaceskip=0pt\relax}
\providecommand{\BIBentryALTinterwordstretchfactor}{4}
\providecommand{\BIBentryALTinterwordspacing}{\spaceskip=\fontdimen2\font plus
\BIBentryALTinterwordstretchfactor\fontdimen3\font minus
  \fontdimen4\font\relax}
\providecommand{\BIBforeignlanguage}[2]{{%
\expandafter\ifx\csname l@#1\endcsname\relax
\typeout{** WARNING: IEEEtran.bst: No hyphenation pattern has been}%
\typeout{** loaded for the language `#1'. Using the pattern for}%
\typeout{** the default language instead.}%
\else
\language=\csname l@#1\endcsname
\fi
#2}}
\providecommand{\BIBdecl}{\relax}
\BIBdecl

\bibitem{ram-hadas}
H.~Kress-{G}azit, T.~Wongpiromsarn, and U.~Topcu, ``Correct, reactive robot
  control from abstraction and temporal logic specifications,'' \emph{IEEE
  Robotics and Automation Magazine}, vol.~18, pp. 65--74, 2011.

\bibitem{itac}
T.~Wongpiromsarn, U.~Topcu, and R.~Murray, ``Receding horizon temporal logic
  planning,'' \emph{IEEE Transactions on Automatic Control}, vol.~57, no.~11,
  pp. 2817--2830, 2012.

\bibitem{GR(1)case}
R.~Bloem, S.~Galler, N.~Piterman, A.~Pnueli, and M.~Weiglhofer, ``Automatic
  hardware synthesis from specifications: A case study,'' in \emph{Proc.
  Design, Automation and Test in Europe}, 2007.

\bibitem{vms0}
T.~Wongpiromsarn, U.~Topcu, and R.~M. Murray, ``Formal synthesis of embedded
  control software for vehicle management systems,'' in \emph{Proc. AIAA
  Infotech\@ Aerospace}, 2011.

\bibitem{pnueli-hardmis}
A.~Pnueli and R.~Rosner, ``On the synthesis of a reactive module,'' in
  \emph{Proc. Symposium on Principles of Programming Languages}, 1989, pp.
  179--190.

\bibitem{piterman}
N.~Piterman, A.~Pnueli, and Y.~Sa'ar, ``{Synthesis of reactive(1) designs},''
  in \emph{Proc. International Conference on Verification, Model Checking, and
  Abstract Interpretation}, vol. 3855, 2006, pp. 364--380.

\bibitem{pnueli2010jtlv}
A.~Pnueli, Y.~SaÕar, and L.~D. Zuck, ``Jtlv: A framework for developing
  verification algorithms,'' in \emph{Computer Aided Verification}.\hskip 1em
  plus 0.5em minus 0.4em\relax Springer, 2010, pp. 171--174.

\bibitem{bloem2010ratsy}
R.~Bloem, A.~Cimatti, K.~Greimel, G.~Hofferek, R.~K{\"o}nighofer, M.~Roveri,
  V.~Schuppan, and R.~Seeber, ``{RATSY}--a new requirements analysis tool with
  synthesis,'' in \emph{Computer Aided Verification}.\hskip 1em plus 0.5em
  minus 0.4em\relax Springer, 2010, pp. 425--429.

\bibitem{wongpiromsarn2011tulip}
T.~Wongpiromsarn, U.~Topcu, N.~Ozay, H.~Xu, and R.~M. Murray, ``{TuLiP}: a
  software toolbox for receding horizon temporal logic planning,'' in
  \emph{Proceedings of the 14th International Conference on Hybrid systems:
  Computation and Control}.\hskip 1em plus 0.5em minus 0.4em\relax ACM, 2011,
  pp. 313--314.

\bibitem{rayna-report}
B.~F.~R. Dimitrova and B.~Finkbeiner, ``Abstraction refinement for games with
  incomplete information,'' Reports of SFB/TR 14 AVACS 43, SFB/TR 14 AVACS,,
  Tech. Rep., October 2008.

\bibitem{apt2011lectures}
K.~R. Apt and E.~Gr{\"a}del, \emph{Lectures in game theory for computer
  scientists}.\hskip 1em plus 0.5em minus 0.4em\relax Cambridge University
  Press, 2011.

\bibitem{chatterjee2006algorithms}
K.~Chatterjee, L.~Doyen, T.~A. Henzinger, and J.-F. Raskin, ``Algorithms for
  omega-regular games with imperfect information,'' in \emph{Computer Science
  Logic}.\hskip 1em plus 0.5em minus 0.4em\relax Springer, 2006, pp. 287--302.

\bibitem{cimatti2012smt}
A.~Cimatti, S.~Mover, and S.~Tonetta, ``{SMT}-based verification of hybrid
  systems.'' in \emph{AAAI}, 2012.

\bibitem{arnold2003games}
A.~Arnold, A.~Vincent, and I.~Walukiewicz, ``Games for synthesis of controllers
  with partial observation,'' \emph{Theoretical computer science}, vol. 303,
  no.~1, pp. 7--34, 2003.

\bibitem{Martin2006}
M.~Wulf, L.~Doyen, and J.-F. Raskin, ``A lattice theory for solving games of
  imperfect information,'' in \emph{Hybrid Systems: Computation and Control},
  ser. Lecture Notes in Computer Science, J.~Hespanha and A.~Tiwari, Eds.\hskip
  1em plus 0.5em minus 0.4em\relax Springer Berlin Heidelberg, 2006, vol. 3927,
  pp. 153--168.

\bibitem{kaelbling1998planning}
L.~P. Kaelbling, M.~L. Littman, and A.~R. Cassandra, ``Planning and acting in
  partially observable stochastic domains,'' \emph{Artificial Intelligence},
  vol. 101, no.~1, pp. 99--134, 1998.

\bibitem{pineau2002high}
J.~Pineau and S.~Thrun, ``High-level robot behavior control using pomdps,'' in
  \emph{AAAI-02 Workshop on Cognitive Robotics}, vol. 107, 2002.

\bibitem{kurniawati2008sarsop}
H.~Kurniawati, D.~Hsu, and W.~S. Lee, ``{SARSOP}: Efficient point-based pomdp
  planning by approximating optimally reachable belief spaces.'' in
  \emph{Robotics: Science and Systems}, 2008, pp. 65--72.

\bibitem{DimitrovaF08}
R.~Dimitrova and B.~Finkbeiner, ``Abstraction refinement for games with
  incomplete information,'' in \emph{FSTTCS}, ser. Dagstuhl Seminar
  Proceedings, R.~Hariharan, M.~Mukund, and V.~Vinay, Eds., vol. 08004.\hskip
  1em plus 0.5em minus 0.4em\relax Internationales Begegnungs- und
  Forschungszentrum fuer Informatik (IBFI), Schloss Dagstuhl, Germany, 2008.

\bibitem{emerson1990temporal}
E.~A. Emerson, ``Temporal and modal logic.'' \emph{Handbook of Theoretical
  Computer Science, Volume B: Formal Models and Sematics (B)}, vol. 995, p.
  1072, 1990.

\bibitem{chatterjee2006algorithm}
K.~Chatterjee, L.~Doyen, T.~A. Henzinger, and J.-F. Raskin, ``Algorithms for
  omega-regular games with imperfect information,'' in \emph{Computer Science
  Logic}, ser. Lecture Notes in Computer Science, Z.~{\'E}sik, Ed.\hskip 1em
  plus 0.5em minus 0.4em\relax Springer, 2006, vol. 4207, pp. 287--302.

\bibitem{smullyan1995first}
R.~M. Smullyan, \emph{First-order logic}.\hskip 1em plus 0.5em minus
  0.4em\relax Courier Dover Publications, 1995.

\bibitem{McMillan2011}
K.~L. McMillan, ``Interpolants from z3 proofs,'' in \emph{Proceedings of the
  International Conference on Formal Methods in Computer-Aided Design}, ser.
  FMCAD '11.\hskip 1em plus 0.5em minus 0.4em\relax Austin, TX: FMCAD Inc,
  2011, pp. 19--27.

\bibitem{henzinger2002lazy}
T.~A. Henzinger, R.~Jhala, R.~Majumdar, and G.~Sutre, ``Lazy abstraction,'' in
  \emph{ACM SIGPLAN Notices}, vol.~37, no.~1.\hskip 1em plus 0.5em minus
  0.4em\relax ACM, 2002, pp. 58--70.

\end{thebibliography}
\bibliographystyle{IEEEtran}
\end{document}